\documentclass[10pt]{llncs}

\usepackage[normalem]{ulem}
\usepackage{graphicx}
\usepackage{amsfonts}
\usepackage{float,amssymb,amsmath,epsfig,psfrag}
\usepackage{mathrsfs}
\usepackage{tikz}

\def\qed{{\hfill $\Box$}}

\def\man#1{m_{#1}}
\def\woman#1{w_{#1}}

\def\inter#1{\{1,\ldots, #1\}}

\def\y#1{y_{#1}}
\def\s#1{s_{#1}}
\def\p#1{p_{#1}}

\def\pair#1#2{\langle \man{#1}, \woman{#2}\rangle}
\def\noLBLpair#1#2{\langle {#1}, {#2}\rangle}
\def\ro#1{\rho_{#1}}
\def\mbest{M_0}

\def\sm#1{M_{#1}}
\def\roset{{\cal{V}}}

\def\preNT#1{{{N}^-}(#1)}
\def\sucNT#1{{{N}^+}(#1)}

\def\sqInter#1{[1, #1]}

\def\NP{$\mathcal{NP}$}
\def\NPC{$\mathcal{NP}$-complete}
\def\SMfamily{$\textsc{F}$}

\def\SMinst{$\mathcal{I}$}
\def\P#1{$\boldsymbol{\pi#1}$}
\def\SSM{$\textsc{SAT-SM}$}
\def\leafSet#1{\textsc{L($#1$)}}
\def\neighborSet#1{\textsc{N($#1$)}}

\DeclareRobustCommand{\rchi}{{\mathpalette\irchi\relax}}
\newcommand{\irchi}[2]{\raisebox{\depth}{$#1\chi$}}
\def\X#1#2{\rchi_{#1}^{#2}}

\newcommand*\circled[1]{\tikz[baseline=(char.base)]{
            \node[shape=circle,draw,inner sep=1.1pt] (char) {\scriptsize{#1}};}}
            
\usepackage{tikz}
\usepackage{pstricks,fancybox,graphicx,graphics}

\usepackage{enumitem}
\usepackage{multirow}
\usepackage{subfigure}

\newtheorem{defi}{Definition}%\newtheorem{pro}{Proof}

\newtheorem{cor}{Corollary}

\title{On the Complexity of Robust Stable Marriage}

\author{
Begum Genc\inst{1}, 
Mohamed Siala\inst{1},  
Gilles Simonin\inst{2},
Barry O'Sullivan\inst{1}\\
}
\institute{
 Insight, Centre for Data Analytics, Department of Computer Science, University College Cork, Ireland\\
 \email{\{begum.genc, mohamed.siala, barry.osullivan\}@insight-centre.org},\\
 \and
 TASC, Institut Mines Telecom Atlantique, LS2N UMR 6004, Nantes, France\\
 \email{gilles.simonin@imt-atlantique.fr}
}

\begin{document}

\maketitle

\begin{abstract}
\textit{Robust Stable Marriage (RSM)} is a variant of the classical \textit{Stable Marriage} problem, where the robustness of a given stable matching is measured by the number of modifications required for repairing it in case an unforeseen event occurs. 
We focus on the complexity of finding an $(a,b)$-supermatch. 
An $(a,b)$-supermatch is defined as a stable matching in which if any $a$ (non-fixed) men/women break up it is possible to find another stable matching by changing the partners of those $a$ men/women and also the partners of at most $b$ others.
In order to show deciding if there exists an $(a,b)$-supermatch is \NPC, we first introduce a SAT formulation that is \NPC~by using Schaefer's Dichotomy Theorem. 
Then, we show the equivalence between the SAT formulation and finding a $(1,1)$-supermatch on a specific family of instances. 
% We also focus on studying the threshold between the cases in $P$ and \NPC~for this problem. 

\end{abstract}

\section{Introduction}
Matching under preferences is a multidisciplinary family of problems, mostly studied by
 the researchers in the field of economics and computer science. 
There are many variants of the matching problems such as College Admission, Hospital/Residents,
 Stable Marriage, Stable Roommates, etc.%~\cite{manlove2013}.
The reader is referred to the book written by Manlove for a comprehensive background on the subject~\cite{manlove2013}.

Current studies in the literature indicate that different robustness notions for different matching problems are being studied. To the best of our knowledge, the very first robustness notion in matching problems is studied on the Geometric Stable Roommates problem~\cite{ARKIN2009219}.
Later on, there appear a few research papers on robustly stable mechanisms in matching
 markets~\cite{THEC:THEC34,DBLP:journals/geb/Afacan12,Drummond:2013:EAS:2540128.2540145}. 
The most recent notions are proposed in Stable Marriage problem, where one of them uses
 a probability model and a social cost function to measure robustness~\cite{Jacobovic2016PerturbationRS},
and the other one uses a cost function to calculate repair costs of each stable matching and
use it as a measure of robustness~\cite{genc17ijcai}.

We work on the robustness notion of stable matching proposed by Genc et. al.~\cite{genc17ijcai}.
In the context of Stable Marriage, 
%a \textit{stable matching} $M$ is 
%defined as a mapping
 %between men and women 
%set of pair such that each men/women are matched with at most one person from the opposite sex and
 %there is no man-woman pair that prefer each other to their situations in $M$.
 the purpose is to find a matching $M$ between men and women such that 
% each person is matched to at most one partner from the opposite sex and 
 no pair $\langle man,woman \rangle $ prefer each other to their situations in $M$.
% has an incentive to deviate from $M$ by being matched together.  
The authors of~\cite{genc17ijcai} introduced the notion of $(a,b)$-supermatch as a measure of robustness. % and refer to finding an $(a,b)$-supermatch as the \textit{Robust Stable Marriage (RSM)} problem. 
An \textit{$(a,b)$-supermatch} 
is a stable matching such that if any $a$ agents (men or woman) break up it is possible to find another stable matching by changing the partners of those $a$ agents with also changing the partners of at most $b$ others.
%They also propose a polynomial-time procedure for measuring the $(1,b)$-robustness of a given problem. 
This notion is inspired by the work of Ginsberg et al. on $(a,b)$-supermodels in
  Boolean Satisfiability~\cite{Ginsberg98supermodelsand}, and $(a,b)$-super solutions in CSP by Hebrard et.al.
 ~\cite{04-hebrard1,04-hebrard2,07-hebrard-phd}.
 Both finding an $(a,b)$-supermodel and an $(a,b)$-super solution are shown to be \NPC.
However, they leave the complexity of this problem open~\cite{genc17ijcai}.

The focus of this paper is to study the complexity of finding an $(a,b)$-supermatch. % according to parameters $a$ and $b$. 
%We are specially interested in studying the threshold between the cases in complexity classes $\mathcal{P}$ and \NPC~for this problem.
In order to show that the general case of RSM, which is the decision of existence of an $(a,b)$-supermatch, is \NPC, it is sufficient to show that a restricted version of the general problem is \NPC. 
Thus, we first show that the decision problem for finding a $(1,1)$-supermatch on a restricted family of instances is \NPC, then we generalize this complexity result to the general case. 

 Figure~\ref{fig:NPCompletenessGeneralization} illustrates the hierarchy between different cases of 
  finding an $(a,b)$-supermatch. 

 \begin{figure}[h!]
     \centering
     \includegraphics[width=.6\textwidth]{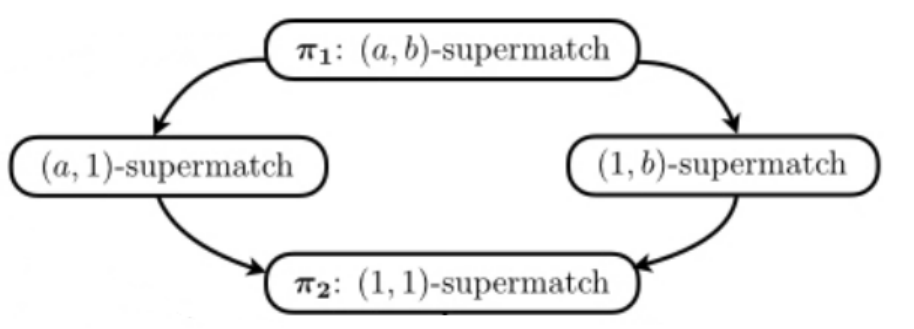}
     \caption{Illustration of the complexity hierarchy between the different cases of RSM.}
     \label{fig:NPCompletenessGeneralization}
 \end{figure}

In Section~\ref{sec:background}, the notations and the basics of the Stable Marriage problem are introduced, 
then we refer to some theorems that are already present in the literature. 
In Section~\ref{sec:specificFamily}, we describe a restricted family of Stable Marriage instances. 
In Section~\ref{sec:complexity}, after defining a specific SAT formulation based on the properties of the restricted family,  
we show by reduction that the decision version of finding an $(1,1)$-supermatch is \NPC. 
%Lastly, in Section~\ref{sec:polynomialCases}, we discuss some polynomial cases for the general problem of Robust Stable Marriage.

\section{Notations \& Background}
\label{sec:background}

% We give in this section the different notations used throughout the paper as well as some background information about the Stable Marriage problem and Schaefer's dichotomy theorem for satisfiability.

% \subsection{Stable Marriage}

An instance of the \textit{Stable Marriage problem (with incomplete lists)} takes as input a set of men $U = \{\man{1}, \man{2},\ldots,\man{n_1} \}$ and a
 set of women $W = \{\woman{1}, \woman{2},\ldots,\woman{n_2} \}$ where each person has an ordinal preference list
 over members of the opposite sex. 
For the sake of simplicity we suppose in the rest of the paper that $n_1=n_2$. 
A \textit{pair} $\langle \man{i}, \woman{j} \rangle$ is acceptable if $\woman{j}$ (respectively $\man{i}$) appears in the preference list of $\man{i}$ (respectively $\woman{j}$). 
A matching is a set of acceptable pairs where each man (respectively woman) appears at most once in any pair of $M$. 
%\todo{Change the definition of matching because the preferences could be incomplete}
%A \textit{\textbf{matching}} $M$ is a one-to-one correspondence between $U$ and $W$. 
%For each man $\man{i}$, $M(\man{i})=\woman{j}$ is called the partner of $\man{i}$, 
%and it is denoted by $M(\woman{j})=\man{i}$ in the other way around.
%$M$ can be considered as a set of pairs. 
%In that case, a pair $\pair{i}{j} \in M$ iff $M(\man{i})=\woman{j}$.  
If $\pair{i}{j} \in M$, we say that $\woman{j}$ (respectively $\man{i}$) is the partner of $\man{i}$ (respectively $\woman{j}$) and then we denote $M(\man{i})=\woman{j}$ and $M(\woman{j})=\man{i}$. 
A pair $\pair{i}{j}$ (sometimes denoted as $\noLBLpair{i}{j}$) is said to be \textit{blocking} a matching $M$ if $\man{i}$ prefers $\woman{j}$ to $M(\man{i})$
 and $\woman{j}$ prefers $\man{i}$ to $M(\woman{j})$. 
A \textit{matching} $M$ is called \textit{stable} if there exists no blocking pair for $M$.  
A \textit{pair} $\pair{i}{j}$ is said to be \textit{stable} if it appears in a stable matching. 
The stable matching, in which each man is matched to their most preferred stable partner is called the man-optimal matching and denoted by $\mbest$.
A pair $\pair{i}{j}$ is called \textit{fixed} if $\pair{i}{j}$ appears in every stable matching. In this case, the man $m_{i}$ and woman $w_{j}$ are called fixed.
%The man-optimal matching is denoted by $\mbest$ and the woman-optimal (man-pessimal) matching is denoted by $\wbest$.
In the rest of the paper we use $n$ to denote the number of non-fixed men and \SMinst~ to denote an instance of a Stable Marriage problem.
We measure the distance between two stable matchings $\sm{i}, \sm{j}$ by the number of men that have different partners in $\sm{i}$ and $\sm{j}$, denoted by $d(\sm{i}, \sm{j})$.
 
% A stable matching $M$ is called an \textbf{\textit{(a,b)}-supermatch} if any $a$ agents decide to break their matches in $M$, thereby breaking $a$ pairs, it is possible to ``repair'' $M$ (i.e., find another stable matching) by changing the assignments of those $a$ agents and the assignments of at most $b$ others.  
Formally, a stable matching $M$ is said to be $(a,b)$-supermatch if for any set $\Psi \subset M$ of $a$ stable pairs that are not fixed, there exists a stable matching $M'$ such that 
$M' \cap \Psi = \emptyset$ and $d(M, M') -a  \leq b$ ~\cite{genc17ijcai}.	

\begin{defi}[\P{_1}]
\label{def:pi}
INPUT: $a,b \in \mathbb{N}$, and a Stable Marriage instance \SMinst.\\
QUESTION: Is there an $(a,b)$-supermatch for \SMinst?
%stable matching $M$ in \SMinst~such that if any $\boldsymbol{a}$ pairs want to break their match in $M$, it is guaranteed to have another stable matching in \SMinst\ by changing the partners of these $\boldsymbol{a}$ pairs and also the partners of at most $\boldsymbol{b}$ other pairs in $M$?
\end{defi}

%The structure that represents all stable matchings forms a \textit{lattice} $\mathscr{M}$.
%In this lattice, the man-optimal matching is denoted by $\mbest$ and the woman-optimal (man-pessimal) matching is denoted by $\wbest$. 

Let $M$ be a stable matching. 
A \textit{\textbf{rotation}} $\ro{} = (\pair{k_0}{k_0}, \pair{k_1}{k_1}, \ldots,$
 $\pair{k_{l-1}}{k_{l-1}})$ (where $l \in \mathbb{N}^*$) is an ordered list of pairs
 in $M$ such that changing the partner of each man $\man{k_i}$ to the partner of the next man
 $\man{k_{i+1}}$ (the operation +1 is modulo $l$) in the list $\ro{}$ leads to a stable matching denoted by $M / \ro{}$.
The latter is said to be obtained after
\textit{eliminating} $\ro{}$ from $M$.
In this case, we say that $\pair{l_i}{l_i}$ is \textit{eliminated} by $\ro{}$, whereas
 $\pair{l_i}{l_{i+1}}$ is \textit{produced} by $\ro{}$, and that $\ro{}$ is
 \textit{exposed} on $M$.
If a pair $\pair{i}{j}$ appears in a rotation $\ro{}$, we denote it by $\pair{i}{j} \in \ro{}$.
Additionally, if a man $\man{i}$ appears at least in one of the pairs in the rotation $\ro{}$,
 we say $\man{i}$ is \textit{involved} in $\ro{}$.
%Note that, it is always the case that $M$ (strictly) dominates $M / \ro{}$. 
There exists a partial order for rotations. 
A rotation  $\rho'$ is said to precede another rotation $\rho$ (denoted by $\rho'  \prec\prec  \rho $),
 if $\rho'$ is eliminated in every sequence of eliminations that starts at $\mbest$ and ends at a stable
 matching in which $\rho$ is exposed~\cite{Gusfield:1989:SMP:68392}.
Note that this relation is transitive, that is, $\rho''  \prec\prec \rho' \wedge \rho'  \prec\prec
  \rho  \implies \rho'' \prec\prec  \rho $. 
Two rotations are said to be \textit{incomparable} if none of them precede the other.  

The structure that represents all rotations and their partial order is a directed graph called
 \textit{\textbf{rotation poset}} denoted by $\Pi = (\roset, E)$. 
Each rotation corresponds to a vertex in $\roset$ and there exists an edge from $\rho'$ to $\rho$
 if $\rho'$ precedes $\rho$. 
% The number of rotations in $\Pi$ is bounded by $n(n-1)/2$ and the number of arcs is bounded by $n^2$. 
% It should also be noted that the construction of $\Pi$ can be done in $O(n^2)$~\cite{Gusfield:1989:SMP:68392}.
There are two different edge types in a rotation poset: \textbf{\textit{type 1}} and \textbf{\textit{type 2}}. 
Suppose $\pair{i}{j}$ is in rotation $\ro{}$, 
if $\ro{}'$ is the unique rotation that moves $\man{i}$ to $\woman{j}$ then $(\ro{}', \ro{}) \in E$ and $\ro{}'$ is called
 a type 1 predecessor of $\ro{}$. 
If $\ro{}$ moves $\man{i}$ below $\woman{j}$, and $\ro{}' \not= \ro{}$ is the unique rotation that moves $\woman{j}$ above $\man{i}$, then
 $(\ro{}', \ro{}) \in E$ and $\ro{}'$ is called a type 2 predecessor of $\ro{}$~\cite{Gusfield:1989:SMP:68392}. 
A node that has no outgoing edges is called a \textit{leaf node} and a node that has no incoming edges is called \textit{root node}.
 
A \textit{\textbf{closed subset}} $S$ is a set of rotations such that for any rotation $\rho$ in $S$, if
 there exists a rotation $\rho'$ that precedes $\rho$ then $\rho'$ is also in $S$.
Every closed subset in the rotation poset corresponds to a stable matching~\cite{Gusfield:1989:SMP:68392}.
Let \leafSet{S} be the set of rotations that are the leaf nodes of the graph induced by the rotations in $S$.
Similarly, let \neighborSet{S} be the set of the rotations that are not in $S$, but all of their predecessors are in $S$. 
This can be illustrated as having a cut in the graph $\Pi$, where the cut divides $\Pi$ into two sub-graphs, namely $\Pi_1$ and $\Pi_2$. If there are any comparable nodes between $\Pi_1$ and $\Pi_2$, $\Pi_1$ is the part that contains the preceding rotations. Eventually, $\Pi_1$ corresponds to the closed subset $S$, \leafSet{S} corresponds to the leaf nodes of $\Pi_1$ and \neighborSet{S} corresponds to the root nodes of $\Pi_2$. 

An important remark is that there is a 1-1 correspondence between the matchings in \SMinst~ and the sets of incomparable rotations in \roset.
A closed subset is defined by adding all predecessors of each node in the subset to the subset. Equivalently, if all rotations that precede some other rotations in $S$ are removed from $S$, the resulting set corresponds to a set of incomparable nodes, namely \leafSet{S}. 

% \begin{corollary}
% \!Each set of incomparable rotations in $\roset$\! corresponds to a stable matching of\ \SMinst.
% \label{cor:incompSM}
% \end{corollary}
% %\vspace{-8mm}
% \begin{proof}
% %A closed subset is defined by adding all predecessors of each node in the subset to the subset. Equivalently, if all rotations that precede some other rotations in $S$ are removed from $S$, the resulting set corresponds to a set of incomparable nodes, namely \leafSet{S}. 
% In a closed subset $S$, it is always true that for every rotation $\rho$ in $S$, all the predecessors of $\rho$ are in $S$. 
% Therefore, each set of incomparable nodes can be used to define a stable matching $\sm{}$, where the closed subset of $\sm{}$ is obtained by adding all predecessors of each rotation to the set. 
% \qed
% \end{proof} 

Predecessors of a rotation $\ro{}$ in a rotation poset are denoted by $\preNT{\ro{}}$ and successors are denoted by $\sucNT{\ro{}}$.
% We measure the distance between two stable matchings $\sm{i}, \sm{j}$ by the number of men that have different partners in $\sm{i}$ and $\sm{j}$, denoted by $d(\sm{i}, \sm{j})$.
We also denote by $X(R)$ the set of men involved in a set of rotations $R$.

Let us illustrate these terms on a sample SM instance specified by the preference lists of 7 men/women in Table~\ref{table:sm} given by Genc et. al~\cite{genc17ijcai}. 
For the sake of clarity, each man $\man{i}$ is denoted with $i$ and each woman $\woman{j}$ with $j$.
Figure~\ref{figure:closedSubset} represents the rotation poset and all the rotations associated with this sample. 
\begin{table}[ht]
\begin{minipage}[b]{0.42\linewidth}
\centering
\begin{tabular}{|l|l|l|l|l|l|l|l|l|l|l|l|l|l|l|l|l|}
\cline{1-8} \cline{10-17}
$m_0$ & \multicolumn{7}{l|}{0 6 5 2 4 1 3} &  & $w_0$ & \multicolumn{7}{l|}{2 1 6 4 5 3 0} \\ \cline{1-8} \cline{10-17} 
$m_1$ & \multicolumn{7}{l|}{6 1 4 5 0 2 3} &  & $w_1$ & \multicolumn{7}{l|}{0 4 3 5 2 6 1} \\ \cline{1-8} \cline{10-17} 
$m_2$ & \multicolumn{7}{l|}{6 0 3 1 5 4 2} &  & $w_2$ & \multicolumn{7}{l|}{2 5 0 4 3 1 6} \\ \cline{1-8} \cline{10-17} 
$m_3$ & \multicolumn{7}{l|}{3 2 0 1 4 6 5} &  & $w_3$ & \multicolumn{7}{l|}{6 1 2 3 4 0 5} \\ \cline{1-8} \cline{10-17} 
$m_4$ & \multicolumn{7}{l|}{1 2 0 3 4 5 6} &  & $w_4$ & \multicolumn{7}{l|}{4 6 0 5 3 1 2} \\ \cline{1-8} \cline{10-17} 
$m_5$ & \multicolumn{7}{l|}{6 1 0 3 5 4 2} &  & $w_5$ & \multicolumn{7}{l|}{3 1 2 6 5 4 0} \\ \cline{1-8} \cline{10-17} 
$m_6$ & \multicolumn{7}{l|}{2 5 0 6 4 3 1} &  & $w_6$ & \multicolumn{7}{l|}{4 6 2 1 3 0 5} \\ \cline{1-8} \cline{10-17} 
\end{tabular}
 \bigskip
 \caption{Preference lists for men (left) and women (right) for a sample instance of size 7.}
 \label{table:sm}
\end{minipage}\hfill
\begin{minipage}[b]{0.42\linewidth}
 \centering
 \includegraphics[width=\textwidth]{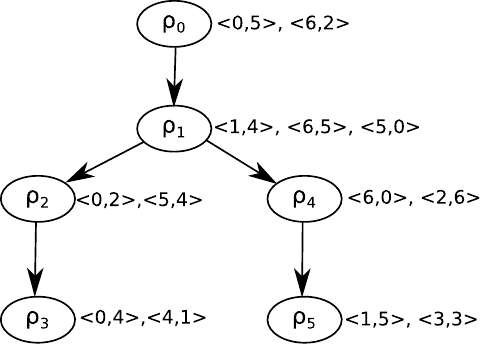}
 \caption{Rotation poset of the instance given in Table~\ref{table:sm}.}
 \label{figure:closedSubset}
\end{minipage}
\end{table}

%\vspace{-5mm}
%The man-optimal matching for the given sample is $\mbest = \{ (0,5), (1,4), (2,6)$, $(3,3), (4,1), (5,0), (6,2) \}$, where the corresponding closed subset is $S_0 = \emptyset$. 
%The closed subset $S_1 = \{ \ro{0} \}$ corresponds to the stable matching obtained by exposing $\ro{0}$ on $\mbest$, $ \mbest/ \ro{0} = \sm{1} = \{ (0,2), (1,4), (2,6), (3,3), (4,1), (5,0), (6,5) \}$.
%Similarly, 

All the stable matchings of the instance given in Figure~\ref{figure:closedSubset} are listed in Table~\ref{table:allMatchingsListed}.
For the sake of example, the stable matching $\sm{2}$ has the corresponding closed subset $S_2$, which is composed of the rotations $\{ \ro{0}, \ro{1} \} $, and it is obtained by exposing the rotation $\ro{1}$ on $\sm{1}$, also denoted by $ \sm{2} = \sm{1} / \ro{1} = \{ (0,2), (1,5)$, $(2,6), (3,3), (4,1), (5,4), (6,0) \}$. For $\sm{2}$, the leaf and neighbor nodes can be identified as \leafSet{S_2} = $ \{ \ro{1} \} $ and \neighborSet{S_2} = $\{ \ro{2}, \ro{4} \}$.

\begin{table}[h!]
\centering
\begin{tabular}{|l|c|}
\hline
\textbf{Stable Matching}                             & \textbf{Pairs}                                                  \\ \hline
$\sm{0}$                                      & $\{ (0, 5), (1, 4), (2, 6), (3, 3), (4, 1), (5, 0), (6, 2) \} $ \\ \hline
$\sm{1} = \sm{0} / \ro{0}$                    & $\{ (0, 2), (1, 4), (2, 6), (3, 3), (4, 1), (5, 0), (6, 5) \} $ \\ \hline
$\sm{2} = \sm{1} / \ro{1}$                    & $\{ (0, 2), (1, 5), (2, 6), (3, 3), (4, 1), (5, 4), (6, 0) \} $ \\ \hline
$\sm{3} = \sm{2} / \ro{4}$                    & $\{ (0, 2), (1, 5), (2, 0), (3, 3), (4, 1), (5, 4), (6, 6) \} $ \\ \hline
$\sm{4} = \sm{3} / \ro{5}$                    & $\{ (0, 2), (1, 3), (2, 0), (3, 5), (4, 1), (5, 4), (6, 6) \} $ \\ \hline
$\sm{5} = \sm{2} / \ro{2}$                    & $\{ (0, 4), (1, 5), (2, 6), (3, 3), (4, 1), (5, 2), (6, 0) \} $ \\ \hline
$\sm{6} = \sm{5} / \ro{4} = \sm{3} / \ro{2}$  & $\{ (0, 4), (1, 5), (2, 0), (3, 3), (4, 1), (5, 2), (6, 6) \} $ \\ \hline
$\sm{7} = \sm{4} / \ro{2} = \sm{6} / \ro{5}$  & $\{ (0, 4), (1, 3), (2, 0), (3, 5), (4, 1), (5, 2), (6, 6) \} $ \\ \hline
$\sm{8} = \sm{5} / \ro{3}$                    & $\{ (0, 1), (1, 5), (2, 6), (3, 3), (4, 4), (5,2 ), (6, 0) \} $ \\ \hline
$\sm{9} = \sm{6} / \ro{3} = \sm{8} / \ro{4}$  & $\{ (0, 1), (1, 5), (2, 0), (3, 3), (4, 4), (5, 2), (6, 6) \} $ \\ \hline
$\sm{10} = \sm{7} / \ro{3} = \sm{9} / \ro{5}$ & $\{ (0, 1), (1, 3), (2, 0), (3, 5), (4, 4), (5, 2), (6, 6) \} $ \\ \hline
\end{tabular}
\caption{The list of all stable matchings for the instance given in Figure~\ref{figure:closedSubset}.}
\label{table:allMatchingsListed}
\end{table}

\subsection{Schaefer's Dichotomy Theorem for Satisfiability}

The original Schaefer's dichotomy theorem is proposed in~\cite{Schaefer:1978:CSP:800133.804350}. 
In this section, we use the same terminology and notations as in~\cite{ubcomplexity}. 
A \textbf{literal} is a Boolean variable or its negation. 
A \textbf{clause} is a disjunction of literals. 
If $x$ is a Boolean variable, then the literal $x$ is called \textit{positive} and the literal $\neg{x}$ is called \textit{negative}.
We shall use the term \textit{formula} to say a Boolean formula given in a conjunctive normal form (\textbf{CNF}) as a finite set of clauses. 

A formula is called \textbf{Horn} (respectively \textbf{dual-Horn}) if every clause in this formula contains a positive (respectively negative) literal. 
A \textit{linear equation over the $2$-element field} 
is an expression of the form $x_1 \oplus x_2 \ldots \oplus x_k =\delta $
where $\oplus$ is the sum modulo 2 operator and $\delta$ is $0$ or $1$.
An \textit{\textbf{affine formula}} is a 
conjunction of linear equations over the $2$-element field. 
%\end{defi}

An \textbf{assignment} is a mapping from (Boolean) variables to $\{ true, false\} $. An assignment $A$ is said to satisfy a clause $C$ if and only if there exists a variable $x$ such that $C$ contains $x$ and the assignment of $x$ by $A$ is $true$ or 
$C$ contains $\neg{x}$ and the assignment of $x$ by $A$ is $false$. 
A \textit{Boolean Constraint of arity $k$} is a function 
$\phi: \{true,false\} ^k \rightarrow \{true, false\}$.
Let $(x_1, \ldots x_k)$ be a sequence of Boolean variables and 
$\phi$ be a Boolean constraint of arity $k$.
The pair $\langle \phi , (x_1, \ldots x_k) \rangle$ is called 
a \textit{constraint application}.
An assignment $A$ to $(x_1, \ldots x_k) $ satisfies 
$\langle \phi , (x_1, \ldots x_k) \rangle$  
if $\phi$ evaluates to $true$ on the truth values assigned by $A$. 
Let $\Phi$ be a set of constraint applications, and $A$ be an assignment to all variables occurring in $\Phi$. 
$A$ is said to be a \textit{satisfying assignment} of $\Phi$ if $A$ satisfies every constraint application in $A$.

Let $\cal C$ be a set of Boolean constraints. 
$\boldsymbol{SAT({\cal C})}$ is defined as the following decision problem:
Given a finite set $\Phi$ of constraints applications from $\cal C$, is there a satisfying assignment for $\Phi$?

\begin{theorem}{Dichotomy Theorem for Satisfiability~\cite{ubcomplexity,Schaefer:1978:CSP:800133.804350}}. 
Let $\cal C$ be a set of Boolean constraints. 
If $\cal C$ satisfies at least one of the conditions (a)-(f) below, then $SAT({\cal C})$ is in {$\mathcal{P}$}. Otherwise, $SAT({\cal C})$ is \NPC.

\begin{itemize}
\item [a)] Every constraint in $\cal C$ evaluates to $true$ if all assignments are $true$.
\item [b)] Every constraint in $\cal C$ evaluates to $true$ if all assignments are $false$.
\item [c)] Every constraint in $\cal C$ can be expressed as a Horn formula.
\item [d)] Every constraint in $\cal C$ can be expressed as a dual-Horn formula.
\item [e)] Every constraint in $\cal C$ can be expressed as affine formula.
\item [f)] Every constraint in $\cal C$ can be expressed as a 2-CNF formula.
\end{itemize}
\label{theo:SchaeferDichotomy}
\end{theorem}

% second section
\section{A specific problem family} %\gil{Put the description of properties}
\label{sec:specificFamily}

In this section, we describe a restricted, specific family \SMfamily\ of Stable Marriage instances over properties on its generic rotation poset $\Pi_F = (\roset_F,E_F)$.

\begin{enumerate}[label=Property \arabic*,leftmargin=*]
  \item\label{prop:twopairs} Each rotation $\ro{i} \in \roset_F$, contains exactly 2 pairs $\ro{i} = ( \pair{i1}{i1},$ $\pair{i2}{i2} )$.
  \item\label{prop:2pre2suc} Each rotation $\ro{i} \in \roset_F$, has at most 2 predecessors and at most 2 successors. 
  \item\label{prop:type1} Each edge $e_i \in E_F$, is a type 1 edge.
  \item\label{prop:atleast2} For each man $ \man{i}, i \in \sqInter{n}$, $\man{i}$ is involved in at least 2 rotations.
\end{enumerate}

 Figure~\ref{fig:piFamilyGeneral2} illustrates these properties. Note that, the ordering of the pairs is not important as there exist only two pairs in each rotation.

 \begin{figure}[h!]
     \centering
     \includegraphics[width=.87\textwidth]{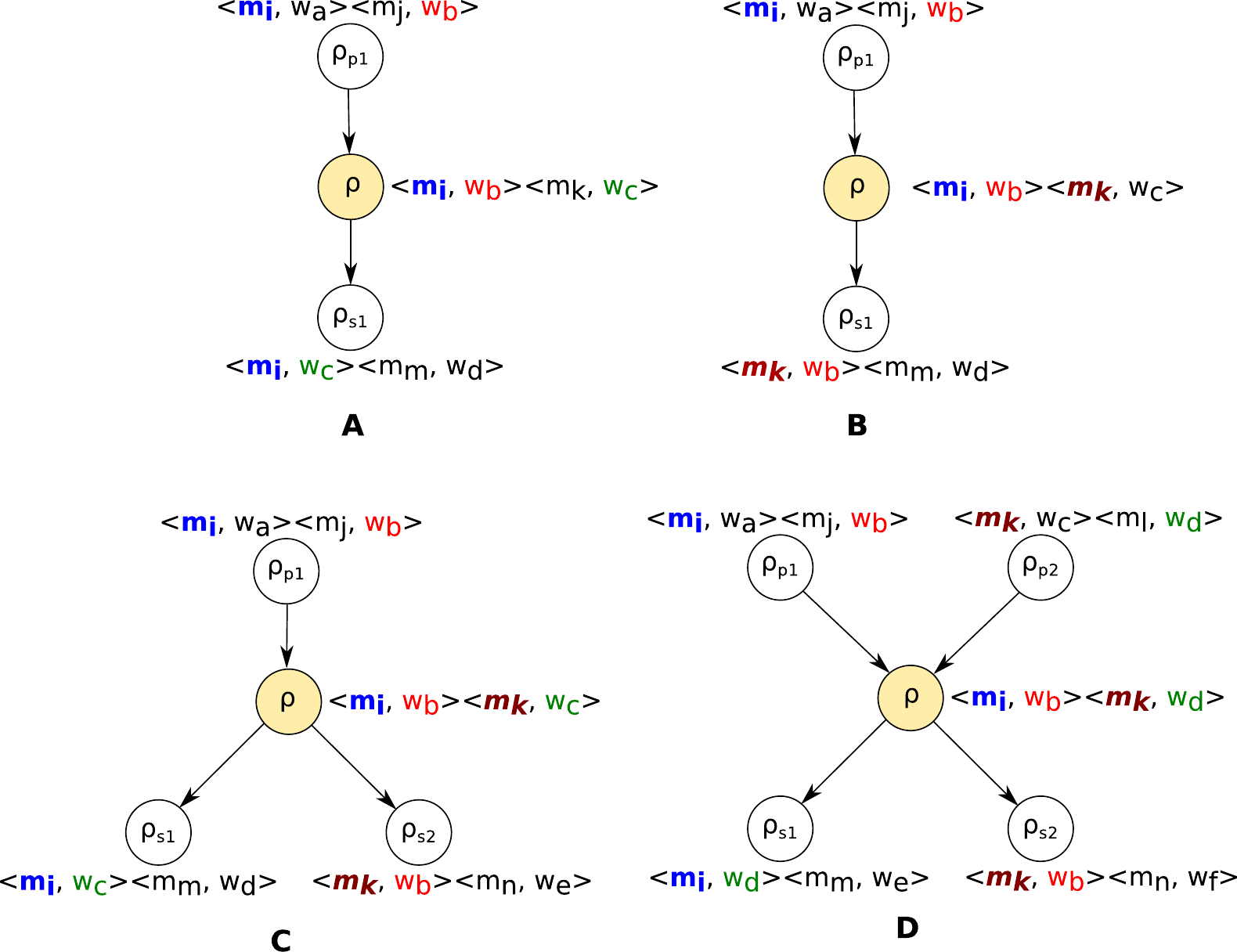}
     \caption{An illustration of the men/women included in rotations for cases where a rotation has exactly 1 predecessor and 1 successor (A,B), 1 predecessor and 2 successors, which is also similar to 1 successor and 2 predecessors (C), 2 predecessors and 2 successors (D).}
     \label{fig:piFamilyGeneral2}
 \end{figure}
 
 We would like to remark the difference between the cases A and B in Figure~\ref{fig:piFamilyGeneral2}.
 Due to~\ref{prop:type1}, any two rotations that have an edge between, contain a man and a woman in common. 
 When the case is generalized to 3 rotations $\ro{p1} - \ro{} - \ro{s1}$, it should be noted that those three rotations either contain the same man $\man{i}$ but contain different women, $\woman{b}$ between $\ro{p1}$ and $\ro{}$, but $\woman{c}$ between $\ro{}$ and $\ro{s1}$ as in the case A or the same woman in all three rotations but different men $\man{i}$ between $\ro{p1}$ and $\ro{}$, but $\man{k}$ between $\ro{}$ and $\ro{s1}$ as in the case B.

\begin{lemma}
For each two different paths $P_1$ and $P_2$ defined on $\Pi_F$, where both start at rotation $\ro{s}$, end at $\ro{t}$, and the pair $\pair{e}{f} \in \ro{s}$, if all rotations on $P_1$ (respectively $P_2$) contain $\man{e}$, at least one of the rotations on $P_2$ (respectively $P_1$) does not contain $\woman{f}$.  
\label{lem:propRematch}
\end{lemma}
\begin{proof}
Suppose for contradiction that $\man{e}$ is involved in all rotations on $P_1$ and $\woman{f}$ is involved in all rotations on $P_2$. This scenario is likely to occur as the~\ref{prop:type1} indicates that all edges are of type 1, which is also easy to observe on the Figure~\ref{fig:piFamilyGeneral2}, any two rotations connected by an edge in the rotation poset always contain a man, and a woman in common. In this case, because of the supposition, the $\man{e}$ and $\woman{f}$ are carried on to the rotation $\ro{t}$, the pair $\pair{e}{f}$ is reproduced. In other words, exposing $\ro{t}$ on a stable matching produces the pair $\pair{e}{f}$. However, this pair is already eliminated by $\ro{s}$, meaning this couple is already produced. The supposition contradicts the fact that exposing rotations on stable matchings causes men to be matched with their less preferred partners, and if a couple is eliminated once can not be produced again.
\end{proof}

\begin{defi}[\P{_1^F}]
\label{def:pi'}
A particular case of \P{_1}, with the restrictions from problem family \SMfamily. 
\end{defi}

% In order to prove that the general problem \P{_1} is \NPC, we first show that the restricted family problem \P{_1^F} is \NPC. 
% In order to do this, we prove it for a particular case noted \P{_2^F}.

% BEGUM: added
\begin{defi}[\P{_2}]
\label{def:pi_2}
The special case of \P{_1}, where $a=1, b=1$. 
% Let \P{_2} be the special case of \P{_1}, where $a=1, b=1$. 
\end{defi}

\begin{defi}[\P{_2^F}]
\label{def:pi2}
% Let \P{_2^F} be finding a $(1,1)$-supermatch problem with the restriction from \SMfamily. \\
INPUT: A Stable Marriage instance \SMinst~from family \SMfamily.\\
QUESTION: Is there a $(1,1)$-supermatch for \SMinst?
%stable matching $M$ in \SMinst~such that if any one of the pairs $\pair{i}{j}$ want to break their match in $M$, it is guaranteed to have another stable matching in \SMinst\ by changing the partners of $\man{i}$ and $\woman{j}$ and also the partners of at most $1$ other pair $\pair{k}{l} \in M$, where $i \neq k, j \neq l$.
\end{defi}

\section{Complexity results}
\label{sec:complexity}

In order to prove that the general problem \P{_1} is \NPC, we show that the restricted family problem \P{_1^F} is \NPC. 
And in order to do this, we first prove it for a particular case noted \P{_2^F} by introducing an \NPC~
SAT formulation denoted by \SSM.

\subsection{\NP -completeness}

%\begin{defi}
%\label{def:SSM}
%Let \SSM\ be the specific SAT SuperMatch problem that we will define in this section. 
\SSM\ takes as input a set of integers $\boldsymbol{\X{}{}}=[1,|\boldsymbol{\X{}{}}|]$, 
$n, n \in \mathbb{N^*}$ lists $l_1, l_2, \ldots, l_n$
where each list is an ordered list of integers of $\boldsymbol{\X{}{}}$,
and three sets of distinct Boolean variables
$Y= \{\y{e} ~|~ e \in \boldsymbol{\X{}{}} \}$, 
$S= \{\s{e} ~|~ e \in \boldsymbol{\X{}{}} \}$, 
and  $P= \{\p{e} ~|~ e \in \boldsymbol{\X{}{}} \} \}$.

%defined by a set $\boldsymbol{\X{}{}}$ of integers, 
%$n$ lists denoted by $l_a \in \boldsymbol{L}$, where $a \in \sqInter{n}$ containing the integers in $\boldsymbol{\X{}{}}$ in an orderly manner, 
%and a specific SAT formula built from these lists. 
%The \textbf{decision problem} consists in deciding if there exists a valid assignment of each literal that satisfies the CNF formula.
\paragraph*{Conditions on the lists: }
The lists $l_1, \ldots, l_n$ are subject to the following constraints: 
%For each $a \in \sqInter{n}$, we have one list of ordered indexes 
First, each list $\forall a \in \sqInter{n}$, $l_a = (\X{1}{a}, \ldots, \X{k_{l_a}}{a})$, where $k_{l_a}= |l_a| \geq 2$. 
Second, each element of $\boldsymbol{\X{}{}}$ appears in exactly two different lists. 
%$l_a$ and $l_b$. 
%That is, , where $a,b \in \sqInter{n}$ s.t. $\X{i}{a}= \X{j}{b}$, 
%with $1\leq i \leq k_{l_a}, 1\leq j \leq k_{l_b}$. 
For illustration, the set $\boldsymbol{\X{}{}}$ represents the indexes of rotations and a list $l_a$ represents the index of each rotation having the man $\man{a}$. 
The order in $l_a$ specifies the path in the rotation poset from the first rotation to the last one for a man $\man{a}$. 
And the restriction for having each index in two different lists is related to~\ref{prop:twopairs}.

In addition to those two conditions, we have the following rule over the lists:\\
\phantomsection
\label{prop:pathRestrictionInSAT}\textbf{[Rule 1]}
% For any couple of indexes denoted as $(\X{i}{a}, \X{j}{a})$ from the same list $l_a$ where $a \in \sqInter{n}, j>i$, 
% there cannot exist a sequence between these two indexes, 
% constructed by iterating the two consecutive rules $\sigma$ and $\theta$ for each two consecutive elements (i.e. [$\theta$-$\sigma$-$\theta$- \ldots - $\sigma$(or $\theta$)]). 
% For each list $l_a, a \!\in\! \sqInter{n}$, $1 \leq i \leq k_{l_a}$, each element of the sequence of form $\X{i}{a}$ is followed by:
%For any two elements $\X{i}{m}$ and $\X{j}{m}$ from the same list $l_m$ where $m \in \sqInter{n}, j>i$,  
%there does not exist any sequence that starts at $\X{i}{m}$ and ends at $\X{j}{m}$, 
%constructed by iterating the two consecutive rules $\sigma$ and $\theta$ detailed as follows: 
%An element of the sequence denoted by $\X{i}{a}$ is only followed by one of the following:
%\begin{itemize}
%	\item[$\sigma$)] $\X{i+1}{a}$, where $i+1 \leq k_{l_a}$.
%   	\item[$\theta$)] $\X{j}{b}$, where $\X{i}{a} = \X{j}{b}, b \in \sqInter{n},  a \neq b, 1 \leq j \leq k_{l_b}$.
%\end{itemize}
For any $\X{i}{m}$ and $\X{j}{m}$ from the same list $l_m$ where $m \in \sqInter{n}$ and $ j>i$,  
there does not exist any sequence $S$ that starts at $\X{i}{m}$ and ends at $\X{j}{m}$ 
constructed by iterating the two consecutive rules $\sigma$ and $\theta$ below: 
\begin{itemize}
	\item given $\X{e}{a}\in S$, the next element in $S$ is $\X{e+1}{a}$, where $e+1 \leq k_{l_a}$.
   	\item given $\X{e}{a}\in S$, the next element in $S$ is $\X{f}{b}$, where $\X{e}{a} = \X{f}{b}, a \neq b  \in \sqInter{n}$, and $1 \leq f \leq k_{l_b}$.
\end{itemize}

%Now we can define the SAT formulation. First, we start with defining the variables:  
%Let $\vert \boldsymbol{\X{}{}} \vert \in \mathbb{N}^*$, and
%$Y= \{\y{e} ~|~ e \in \boldsymbol{\X{}{}} \}$, 
%$S= \{\s{e} ~|~ e \in \boldsymbol{\X{}{}} \}$, 
%and  $P= \{\p{e} ~|~ e \in \boldsymbol{\X{}{}} \} \}$ be three sets of distinct boolean variables.

\paragraph*{Conditions on the clauses:}
The CNF that defines \SSM\ is a conjunction of four groups of clauses: \circled{A}, \circled{B}, \circled{C} and \circled{D}.
The groups are subject to the following conditions:

\circled{A}: For any list $l_a = (\X{1}{a}, \ldots, \X{k_{l_a}}{a})$, 
we have a disjunction between the $Y$-elements and the $P$-elements as %$\small{\bigwedge_{l=1}^{n}\left( \bigvee_{i=1}^{k_{l_a}} y_{\X{i}{a} } \lor p_{\X{i}{a}}\right) }$.
$\bigvee_{i=1}^{k_{l_a}} y_{\X{i}{a} } \lor p_{\X{i}{a}} $.

\begin{eqnarray}
\textnormal{\circled{A} is defined by } \small{\boldsymbol{\bigwedge_{a=1}^{n}\left( \bigvee_{i=1}^{k_{l_a}} y_{\X{i}{a} } \lor p_{\X{i}{a}}\right) }}.
\label{clausesA}
\end{eqnarray}

\circled{B}: For any list $l_a = (\X{1}{a}, \ldots, \X{k_{l_a}}{a})$, 
we have a conjunction of disjunctions between two $S$-elements with consecutive indexes as $\bigwedge_{i=1}^{k_{l_a}-1} s_{\X{i}{a}} \lor \neg s_{\X{i+1}{a}}$.

\begin{eqnarray}
\textnormal{\circled{B} is defined by } \small{\boldsymbol{\bigwedge_{a=1}^{n} \bigwedge_{i=1}^{k_{l_a}-1} s_{\X{i}{a}} \lor \neg s_{\X{i+1}{a}} }}.
\label{clausesB}
\end{eqnarray}

\circled{C}: This group of clauses is split in two. 
For any list $l_a = (\X{1}{a}, \ldots, \X{k_{l_a}}{a})$, 
the first sub-group $C_1$ contains all the clauses defined by the logic formula $\bigwedge_{i=1}^{k_{l_a}-1} y_{\X{i}{a}} \rightarrow s_{\X{i}{a}} \land \neg s_{\X{i+1}{a}}$. 
With a CNF notation, it leads to $\bigwedge_{i=1}^{k_{l_a}-1} (\neg y_{\X{i}{a}} \lor s_{\X{i}{a}}) \land (\neg y_{\X{i}{a}} \lor \neg s_{\X{i+1}{a}}) $. Note that, $C_1$ also covers the special case, when $i = k_{l_a}$.

\begin{eqnarray}
C_1 \textnormal{ is defined by }
\small{\boldsymbol{\bigwedge_{a=1}^{n}\left( \bigwedge_{i=1}^{k_{l_a}} (\neg y_{\X{i}{a}} \lor s_{\X{i}{a}}) \land \bigwedge_{i=1}^{k_{l_a}-1} (\neg y_{\X{i}{a}} \lor \neg s_{\X{i+1}{a}}) \right) }}.
\end{eqnarray}

The second sub-group $C_2$ has three specific cases according to the position of elements in the ordered lists. 
As fixed above, each element of $\boldsymbol{\X{}{}}$ appears in exactly two different lists. 
Thus, for any $e \in \boldsymbol{\X{}{} }$, there exists two lists $l_a$ and $l_b$  such that $\X{i}{a} = \X{j}{b} = e$, 
where $i\in \sqInter{k_{l_a}}$ and $j\in \sqInter{k_{l_b}}$.
For each couple of elements of $\boldsymbol{\X{}{}}$ denoted by $(\X{i}{a}, \X{j}{b})$ that are equal to the same value $e$, 
we define a clause with these elements and the next elements in their lists respecting the ordering: 
$s_{\X{i}{a}} \rightarrow y_{\X{i}{a}} \lor s_{\X{i+1}{a}} \lor s_{\X{j+1}{b}}$. 
With a CNF notation it leads to: 
$(\neg s_{\X{i}{a}} \lor y_{\X{i}{a}} \lor s_{\X{i+1}{a}} \lor s_{\X{j+1}{b}})$.

We add the two specific cases where $\X{i}{a}$ or $\X{j}{b}$,  
or both are the last elements of their ordered lists. 
% The complete formula for the set of clauses $C_2$ is:
% \begin{eqnarray}
% \begin{array}{l}
% C_2\quad \left[
% \begin{array}{l}
% \boldsymbol{\underset{\scriptsize{\begin{array}{c}(\X{i}{a}, \X{j}{b})\ s.t.\ \X{i}{a}= \X{j}{b}\\ i\not=k_{l_a},\ j\not=k_{l_b}\end{array} } } {\bigwedge} \hspace{1.5pc}
% (\neg s_{\X{i}{a}} \lor y_{\X{i}{a}} \lor s_{\X{i+1}{a}} \lor s_{\X{j+1}{b}}) }\\
% \boldsymbol{\underset{\scriptsize{\begin{array}{c}(\X{i}{a}, \X{j}{b})\ s.t.\ \X{i}{a}= \X{j}{b}\\ i\not=k_{l_a},\ j=k_{l_b}\end{array}}}{\bigwedge} \hspace{1.5pc}
% (\neg s_{\X{i}{a}} \lor y_{\X{i}{a}} \lor s_{\X{i+1}{a}} )}\\
% \boldsymbol{\underset{\scriptsize{\begin{array}{c}(\X{k_{l_a}}{a}, \X{k_{l_b}}{b})\ s.t.\ \X{k_{l_a}}{a}= \X{k_{l_b}}{b}\end{array}}}{\bigwedge}
% (\neg s_{\X{k_{l_a}}{a}} \lor y_{\X{k_{l_a}}{a}} )}
% \end{array}
% \right. \\
% \end{array}
% \end{eqnarray}\\
The complete formula for the set of clauses $C_2$ for each two element $(\X{i}{a}, \X{j}{b})\ s.t.\ \X{i}{a}= \X{j}{b}$ is:

\begin{eqnarray}
\begin{array}{l}
C_2\quad \left[
\begin{array}{l}
\boldsymbol{\underset{\scriptsize{\begin{array}{c}i\not=k_{l_a},\ j\not=k_{l_b}\end{array} } } {\bigwedge} \hspace{1.5pc}
\neg s_{\X{i}{a}} \lor y_{\X{i}{a}} \lor s_{\X{i+1}{a}} \lor s_{\X{j+1}{b}} }\\
\boldsymbol{\underset{\scriptsize{\begin{array}{c}i\not=k_{l_a},\ j=k_{l_b}\end{array}}}{\bigwedge} \hspace{1.5pc}
\neg s_{\X{i}{a}} \lor y_{\X{i}{a}} \lor s_{\X{i+1}{a}} }\\
\boldsymbol{\underset{\scriptsize{\begin{array}{c}i=k_{l_a},\ j=k_{l_b}\end{array}}}{\bigwedge}\hspace{1.5pc}
\neg s_{\X{k_{l_a}}{a}} \lor y_{\X{k_{l_a}}{a}} }
\end{array}
\right. \\
\end{array}
\end{eqnarray}\\

\circled{D}:  Similarly to $C_2$, for each couple of elements of $\boldsymbol{\X{}{}}$ denoted by $(\X{i}{a}, \X{j}{b})$ equal to the same value $e$, 
there exists a clause with these elements and the previous elements in their lists respecting the ordering: 
$p_{\X{i}{a}} \leftrightarrow \neg s_{\X{i}{a}} \land s_{\X{i-1}{a}} \land s_{\X{j-1}{b}}$. 
With a CNF notation, it leads to:

$(\neg p_{\X{i}{a}} \lor \neg s_{\X{i}{a}}) \land 
(\neg p_{\X{i}{a}} \lor s_{\X{i-1}{a}} ) \land
(\neg p_{\X{i}{a}} \lor s_{\X{j-1}{b}} ) \land
( s_{\X{i}{a}} \lor \neg s_{\X{i-1}{a}} \lor \neg s_{\X{j-1}{b}}\lor p_{\X{i}{a}} )
$

% By generalizing the formula for any couple, 
% and by adding the two cases where $\X{i}{l_a}$, or $\X{j}{l_b}$,  
% or both are the first elements of their respective lists,  
% the complete formula \circled{D} is described by:

% \begin{eqnarray}
% \begin{array}{l}
% \textnormal{\circled{D}}\quad \left[
% \begin{array}{l}
% \boldsymbol{
% \hspace{3.6pc} \bigwedge_{\scriptsize{e=1}}^{\scriptsize{\vert \X{}{}\vert}} \hspace{2.4pc} (\neg p_e \lor \neg s_e)}\\

% \boldsymbol{\underset{\scriptsize{\begin{array}{cc}(\X{i}{a}, \X{j}{b})\ s.t.\ \X{i}{a}= \X{j}{b}\\ i\not=1,\ j\not=1\end{array}} }{\bigwedge} 
%  (\neg p_{\X{i}{a}} \lor s_{\X{i-1}{a}} ) \land (\neg p_{\X{j}{b}} \lor s_{\X{j-1}{b}} ) }\\

%  \boldsymbol{\underset{\scriptsize{\begin{array}{cc}(\X{i}{a}, \X{j}{b})\ s.t.\ \X{i}{a}= \X{j}{b}\\ i=1,\ j\not=1\end{array}} }{\bigwedge} 
%   (\neg p_{\X{j}{b}} \lor s_{\X{j-1}{b}} ) }\\

% \boldsymbol{\underset{\scriptsize{\begin{array}{cc}(\X{i}{a}, \X{j}{b})\ s.t.\ \X{i}{a}= \X{j}{b}\\ i\not=1,\ j\not=1\end{array} }}{\bigwedge} 
% ( s_{\X{i}{a}} \lor \neg s_{\X{i-1}{a}} \lor \neg s_{\X{j-1}{b}} \lor p_{\X{i}{a}}) }\\
% \boldsymbol{\underset{\scriptsize{\begin{array}{cc}(\X{i}{a}, \X{j}{b})\ s.t.\ \X{i}{a}= \X{j}{b}\\ i=1,\ j\not=1\end{array}}}{\bigwedge} 
% ( s_{\X{i}{a}} \lor \neg s_{\X{j-1}{a}} \lor p_{\X{i}{a}} )}\\
% \boldsymbol{\underset{\scriptsize{\begin{array}{cc}(\X{1}{a}, \X{1}{b})\ s.t.\ \X{1}{a}= \X{1}{b}\end{array}} }{\bigwedge} 
% ( s_{\X{1}{a}} \lor p_{\X{1}{a}} )
% }

% \end{array}
% \right.
% \end{array}
%  \end{eqnarray}

By generalizing the formula for any couple, 
and by adding the two cases where $\X{i}{l_a}$, or $\X{j}{l_b}$,  
or both are the first elements of their respective lists,  
the complete formula \circled{D} for each two element $(\X{i}{a}, \X{j}{b})\ s.t.\ \X{i}{a} = \X{j}{b} = e$ is described by:

\begin{eqnarray}
\begin{array}{l}
\normalsize{\circled{D}}\quad \left[
\begin{array}{l}

\boldsymbol{\underset{\scriptsize{\begin{array}{cc}i\not=1,\ j\not=1\end{array}} }{\bigwedge} 
 (\neg p_{\X{i}{a}} \lor s_{\X{i-1}{a}} ) \land (\neg p_{\X{j}{b}} \lor s_{\X{j-1}{b}} ) \land } \\ 
 \hspace{6pc} \vspace{0.5pc}
 \boldsymbol{ ( s_{\X{i}{a}} \lor \neg s_{\X{i-1}{a}} \lor \neg s_{\X{j-1}{b}} \lor p_{\X{i}{a}}) }\\
\vspace{0.5pc}

\boldsymbol{\underset{\scriptsize{\begin{array}{cc} i=1,\ j\not=1\end{array}} }{\bigwedge} 
  (\neg p_{\X{j}{b}} \lor s_{\X{j-1}{b}} ) \land ( s_{\X{i}{a}} \lor \neg s_{\X{j-1}{a}} \lor p_{\X{i}{a}} ) }\\
\vspace{0.5pc}

\boldsymbol{\underset{\scriptsize{\begin{array}{cc} i=1,\ j=1 \end{array}} }{\bigwedge} 
 s_{\X{1}{a}} \lor p_{\X{1}{a}} }\\
\vspace{0.5pc}

\hspace{1.9pc} \boldsymbol{\wedge \hspace{1.9pc} \neg p_e \lor \neg s_e }

\end{array}
\right.
\end{array}
\end{eqnarray}

To conclude the definition, the full CNF formula of \SSM\ is 
\circled{A} $\land$ \circled{B} $\land\ C_1 \land C_2\ \land$ \circled{D}.
%\end{defi}

\begin{lemma}
There always exist some clauses of minimum length 4 that are defined over positive literals in \textnormal{\circled{A}}.
\label{lem:proofClauseA}
\end{lemma}
\begin{proof}
For any list of ordered elements $l_a \in \{l_1, l_2, \ldots, l_n\}$, the length of each list is defined as $k_{l_a} \geq 2$ in \SSM, which results in \circled{A} having $n$ clauses that have at least 4 positive literals in each. \qed
\end{proof}

\begin{lemma}
There always exist some clauses of length 2 that are defined over two negative literals in \textnormal{\circled{C}}.
\label{lem:proofClauseC}
\end{lemma}
\begin{proof}
The clauses in \circled{C} consists of two groups. The first group is defined over the list of ordered elements. For any two consecutive elements in a list $l_a \in \{l_1, l_2, \ldots, l_n\}$, there exists two clauses:
$\bigvee_{i=1}^{k_{l_a}-1} (\neg y_{\X{i}{a}} \lor s_{\X{i}{a}}) \land (\neg y_{\X{i}{a}} \lor \neg s_{\X{i+1}{a}})$. By definition, the minimum length of an ordered list $l_a$ is $k_{l_a} = 2$ and therefore the minimum-length list yields in 2 clauses of the defined form. Therefore, the first group includes $2 \times \sum_{a \in \inter{n}} (k_{l_a} - 1)$ clauses. \qed

\end{proof}

\begin{lemma}
Any clause defined over only positive literals of size at least two is not affine.
\label{lem:notAffine}
\end{lemma}
\begin{proof}
Any clause C of the given form with $k$ positive literals have $2^k - 1$ valid assignments. The cardinality of an affine relation is always a power of 2~\cite{Schaefer:1978:CSP:800133.804350}. Thus, C is not affine.\qed
\end{proof}
 
The \SSM\ problem is the question of finding an assignment of the Boolean variables that satisfies the above CNF formula.

\begin{theorem}
\label{thm:pi2}
The \SSM\ problem is \NPC.
\end{theorem}

\begin{proof}

 We use Schaefer's dichotomy theorem (Theorem~\ref{theo:SchaeferDichotomy}) to prove that \SSM\ is \NPC~\cite{Schaefer:1978:CSP:800133.804350}. Schaefer identifies six cases, where if any one of them is valid the SAT problem is solved in polynomial time. Any SAT formula that does not satisfy any of those 6 is $\mathcal{NP}$-complete.

 It is easy to see the properties \textit{a}, \textit{d}, and \textit{f} in Schaefer's Dichotomy do not apply to \SSM\ due to Lemma~\ref{lem:proofClauseA}. 
 Similarly, properties \textit{b} and \textit{c} are not satisfiable because of Lemma~\ref{lem:proofClauseC}. 
 The clauses in \circled{A} are defined as clauses over positive literals and it is known that they always exist by Lemma~\ref{lem:proofClauseA}. By applying Lemma~\ref{lem:notAffine} on the clauses in \circled{A}, we infer that property \textit{e} is not applicable either.
 Hence, \SSM\ is \NPC. \qed

\end{proof}

We can now present the main result of the paper.

\begin{theorem}
\label{thm:pi2NPC}
The decision problem \P{_2^F} is \NPC.
\end{theorem}

\begin{proof}
The verification is 
%of robustness of a given stable matching is previously 
shown to be polynomial-time decidable~\cite{genc17ijcai}. Therefore, \P{_2^F} is in \NP. 
We show that \P{_2^F} is $\mathcal{NP}$-complete by
 presenting a polynomial reduction from the \textsc{\SSM}\ problem to \P{_2^F} as follows.

From an instance \SMinst$_{SSM}$ of \SSM, we construct in polynomial time an instance \SMinst\ of \P{_2^F}. 
This means the construction of the rotation poset $\Pi_F = (\roset_F,E_F)$ with all stable pairs in the rotations, and the preference lists.

We first start constructing the set of rotations $\roset_F$ and then proceed by deciding which man is a part of which stable pair in which rotation.
First, $\forall e \in \boldsymbol{\X{}{}} $, we have a corresponding rotation $\ro{e}$. 
Initially, each rotation contains two ``empty'' pairs. 
Second, $\forall l_a, a \in \sqInter{n}, \forall \X{i}{a} \in \sqInter{k_{l_a}}$, we insert $\man{a}$ as the man to the first empty pair in rotation $\ro{\X{i}{a}}$. 
Since $k_{l_a} \geq 2$ from Lemma \ref{lem:proofClauseA}, 
Each man of \P{_2^F} is involved in at least two rotations (satisfying \ref{prop:atleast2}). 

As each $\X{i}{a}$ appears in exactly two different lists $l_a$ and $l_b$, each rotation is guaranteed to contain exactly two pairs involving different men $\man{a}, \man{b}$ (\ref{prop:twopairs}), and to possess at most two predecessors and  two successors in $\Pi_F$ (\ref{prop:2pre2suc}).

For the construction of the set of arcs $E_F$, 
for each couple of elements of $\boldsymbol{\X{}{}}$ denoted by $(\X{i}{a},\X{i+1}{a})$, $a \in \sqInter{n}, \forall i \in \sqInter{k_{l_a}-1}$, 
we add an arc from $\ro{\X{i}{a}}$ to $\ro{\X{i+1}{a}}$. 
Note that this construction, yields in each arc in $E$ representing a type $1$ relationship (\ref{prop:type1} and~\ref{prop:atleast2}). 
Because each arc links two rotations, where exactly one of the men is involved in both rotations. 
Now, in order to complete the rotation poset $\Pi_F$, the women involved in rotations must also be added. 
The following procedure is used to complete the rotation poset: 

\begin{enumerate}
\item For each element $\X{1}{a} \in \boldsymbol{\X{}{}}$, with $a \in \sqInter{n}$, 
let $\ro{\X{1}{a}}$ be the rotation that involves man $\man{a}$. % and not preceded by any other rotation involving $\man{l}$. 
In this case, the partner of $\man{a}$ in $\ro{\X{1}{a}}$ is completed by inserting woman $\woman{a}$, 
so that the resulting rotation contains the stable pair $\pair{a}{a} \in \ro{\X{1}{l}}$.

\item We perform a breadth-first search on the rotation poset from the completed rotations. 
For each complete rotation $\ro{} = (\pair{i}{b}, \pair{k}{d}) \in \roset_F$, 
let $\ro{s1}$ (resp. $\ro{s2}$) be one of the successor of $\ro{}$ and modifying $\man{i}$ (resp. $\man{k}$). 
If $\ro{s1}$ exists, then we insert the woman $\woman{d}$ in $\ro{s1}$ as the partner of man $\man{i}$. 
In the same manner, if $\ro{s2}$ exists, we insert the woman $\woman{b}$ in $\ro{s2}$ as the partner of man $\man{k}$. 
The procedure creates at most two stable pairs $\pair{i}{d}$ and $\pair{k}{b}$ (see the illustration in Figure \ref{fig:piFamilyGeneral2}.D). 
From the fact that each woman $\woman{b}$ appears in the next rotation as partnered with the next man of the current rotation $\ro{}$, 
in the \SSM\ definition it is equivalent to going from $\X{y}{i}$ to $\X{z+1}{k}$ on lists where $\X{y}{i}=\X{z}{k}, y \in \sqInter{n}, z \in \sqInter{n-1}$. 
Thus the path where the woman appears follow a sequence defined as the one in \textbf{[Rule 1]} from the \SSM\ definition. 
By this rule, we can conclude that Lemma~\ref{lem:propRematch} is satisfied.
\end{enumerate}

All along the construction, 
we showed that all the properties required, to have a valid rotation poset from the family \SMfamily, are satisfied. 
Using this process we are adding equal number of women and men in the rotation poset. 
%This leads to have the same number of men and women in our instance \SMinst.

The last step to obtain an instance \SMinst\ of \P{_2^F} is the construction of the preference lists. 
By using the rotation poset created above, we can construct incomplete preference lists for the men and women. 
 Gusfield et. al. define a procedure to show that every finite poset corresponds to a stable marriage instance~\cite{GUSFIELD1987304}. 
 In their work, they describe a method to create the preference lists using the rotation poset. 
 We use a similar approach for creating the lists as detailed below:

\begin{itemize}
\item Apply topological sort on $\roset_F$.
\item For each man $\man{i} \in \sqInter{n}$, insert woman $\woman{i}$ as the most preferred to $\man{i}$'s preference list. 
\item For each woman $\woman{i} \in \sqInter{n}$, insert man $\man{i}$ as the least preferred to $\woman{i}$'s preference list. 
\item For each rotation $\ro{} \in \roset_F$ in the ordered set, for each pair $\pair{i}{j}$ produced by $\ro{}$, insert $\woman{j}$ to the man $\man{i}$'s list in decreasing order of preference ranking. Similarly, place $\man{i}$ to $\woman{j}$'s list in increasing order of preference ranking.
\end{itemize}

The Lemma~\ref{lem:propRematch} imposed on our rotation poset clearly involves that each preference list contains each member of the opposite sex at most once.
To finish, one can observe 
that the instance obtained respects the Stable Marriage requirements and the specific properties from problem family \SMfamily. \\

$\Leftarrow$
Suppose that there exists a solution to an instance \SMinst\  of the decision problem \P{_2^F}. 
Then we have a $(1,1)$-supermatch and its corresponding closed subset $S$. 
As defined in Section~\ref{sec:background}, \leafSet{S} is the set of leaf nodes of the graph induced by the rotations in $S$, \neighborSet{S} the set of nodes such that all their predecessors are in $S$ but not themselves. 
From these two sets, we can assign all the literals in \SMinst$_{SSM}$ as follows:

\begin{itemize}
\item For each rotation $\ro{i} \in$ \leafSet{S}, set $\y{i} = true$. Otherwise, set $\y{i} = false$.
\item For each rotation $\ro{i} \in S$, set $\s{i} = true$. Otherwise, set $\s{i} = false$.
\item For each rotation $\ro{i} \in$ \neighborSet{S}, set $\p{i} = true$. Otherwise, set $\p{i} = false$.
\end{itemize}
 
If $S$ represents a $(1,1)$-supermatch, that means by removing only one rotation present in \leafSet{S} or by only adding one rotation from \neighborSet{S}, any pair of the corresponding stable matching can be repaired with no additional modifications. 
Thus any men must be contained in a leaf or a neighbor node. 
This leads to having for each man one of the literals assigned to true in his list in \SSM. 
Therefore every clause in \circled{A} in Equation~\ref{clausesA} are satisfied. 
Therefore every clause in \circled{A} are satisfied. 

For the clauses in \circled{B} in Equation~\ref{clausesB}, for any man's list the clauses are forcing each $\s{i}$ literal to be true if the next one $\s{i+1}$ is. 
For the clauses in \circled{B}, for any man's list the clauses are forcing each $\s{i}$ literal to be true if the next one $\s{i+1}$ is. 
By definition of a closed subset, from any leaf of $S$, all the preceding rotations (indexes in the lists) must be in $S$. 
And thus every clause in \circled{B} is satisfied. 

As the clauses in \circled{C} altogether capture the definition of being a leaf node of the graph induced by the rotations in $S$, they are all satisfied by \leafSet{S}.
At last, for the clauses in \circled{D}, it is also easy to see that any rotation being in \neighborSet{S} is equivalent to not being in the solution and having predecessors in. 
Thus all the clauses are satisfied.

Thus we can conclude that this assignment satisfy the SAT formula of \SMinst$_{SSM}$.\\

$\Rightarrow$
Suppose that there exists a solution to an instance \SMinst$_{SSM}$ of the decision problem \SSM. 
Thus we have a valid assignment to satisfy the SAT formula of \SMinst$_{SSM}$. 
We construct a closed subset $S$ to solve \SMinst. 
As previously, we use the sets $L(S)$ and $N(S)$, 
then for each literal $y_i$ assigned to true, we put the rotation $\ro{i}$ in $L(S)$. 
We are doing the same for $p_i$ and $s_i$ as above. 

The clauses in \circled{B} enforce the belonging to $S$ of all rotations preceding any element of $S$, 
thus the elements in $S$ form a closed subset. 
To obtain a $(1,1)$-supermatch, we have to be sure we can repair any couple by removing only one rotation present in \leafSet{S} or by only adding one rotation from \neighborSet{S}. 
The clauses in \circled{C} enforce the rotations in $L(S)$ to be without successors in $S$. 
And in the same way the clauses in \circled{D} enforce the rotations in $N(S)$ to not be in $S$ but have their predecessors in the solution. 

Now we just have to check that all the men are contained in at least one rotation from \leafSet{S} $\cup$ \neighborSet{S}. 
By the clauses from \circled{A}, we know that at least one $\y{e}$ or $\p{e}$ for any man $\man{i}$ is assigned to true. 
Thus from this closed subset $S$, we can repair any couple $\pair{i}{j}$ in one modification by removing/adding the rotation having $\man{i}$. 
Since there exists a $1-1$ equivalence between a stable matching and the closed subset in the rotation poset, we have a $(1,1)$-supermatch. 
\qed
\end{proof}

\begin{cor}
\label{cor:proof}
From the Theorem \ref{thm:pi2NPC} and by generality, both decision problems \P{_1} and \P{_2} are \NPC. 
\end{cor}

\section{Concluding Remarks}
We study the complexity of the Robust Stable Marriage (RSM) problem.
In order to show that given a Stable Marriage instance, deciding if there exists an $(a,b)$-supermatch is \NPC, we first
introduce a SAT formulation which models a specific family of Stable Marriage instances. We show that the formulation is \NPC\ by Schaefer's Dichotomy Theorem. 
Then we apply a reduction from this problem to prove the \NPC ness of RSM. %our problem. 
%Then we show equivalence between the Robust Stable Marriage problem and the SAT formulation. 
%We also present some polynomial cases and discuss the existence of $(a,0)$-supermatches. Although we prove that $(2,0)$-supermatches do not exist, we leave the study of generalizing it to $(a,0)$ as future work.

\section{Acknowledgements}
This research has been funded by Science Foundation Ireland (SFI) under Grant Number SFI/12/RC/2289.

%\bibliographystyle{splncs}
%\bibliography{biblio}

\end{document}